\documentclass[letterpaper, 10 pt, conference]{ieeeconf}  % Comment this line out if you need a4paper

\IEEEoverridecommandlockouts                              % This command is only needed if 
\usepackage{graphicx}
\usepackage[T1]{fontenc}
\usepackage[latin1]{inputenc}   
\usepackage{amsfonts,amsmath,amssymb,array}
\usepackage{cases}
\usepackage{url}
\usepackage[]{hyperref}
\usepackage{tikz}
\usepackage{pgfplots}
\pgfplotsset{compat=1.8}
\usepackage{float}
\usepackage{stfloats}
\usepackage[]{placeins}
\usepackage{graphicx}  
\usepackage{algorithm}
\usepackage{algpseudocode}
 
\usepackage{pifont}
\usepackage{nicefrac}
\usepackage{mathtools}
\usepackage{relsize}
\usepackage{siunitx}

\usetikzlibrary{spy} 
\usetikzlibrary{plotmarks}

\definecolor{dark_orange}{rgb}{1.00000,0.55000,0.00000}%
\definecolor{dark_green}{rgb}{0.40000,0.70000,0.40000}%
\definecolor{gold}{rgb}{1.00000,0.84000,0.00000}%
\definecolor{dark_cyan}{rgb}{0,0.8941,0.8941}%
\tikzstyle{invdot} = [circle,minimum size=0mm,inner sep=0pt,outer sep=0pt]

\newcommand{\agentcyan}{\tikz{\node [invdot] at (0,0){};\draw [ultra thick,color=dark_cyan] (0,0.1) -- (.325,.1);}}
\newcommand{\agentgreen}{\tikz{\node [invdot] at (0,0){};\draw [ultra thick,color=dark_green] (0,0.1) -- (.325,.1);}}
\newcommand{\agentred}{\tikz{\node [invdot] at (0,0){};\draw [ultra thick,color=red] (0,0.1) -- (.325,.1);}}
\newcommand{\agentdashred}{\tikz{\node [invdot] at (0,0){};\draw [ultra thick, dashed,color=red] (0,0.1) -- (.325,.1);}}
\newcommand{\agentgold}{\tikz{\node [invdot] at (0,0){};\draw [ultra thick,color=gold] (0,0) -- (.35,0);\node[thick,color=gold,scale=1.25] at (0.175,0) {\pgfuseplotmark{triangle*}};}}
\newcommand{\agentdashgold}{\tikz{\node [invdot] at (0,0){};\draw [ultra thick,color=gold, dashed] (0,0) -- (.5,0);\node[thick,color=gold,scale=1.25] at (0.25,0) {\pgfuseplotmark{triangle*}};}}
\newcommand{\agentblue}{\tikz{\node [invdot] at (0,0){};\draw [ultra thick,color=blue] (0,0.1) -- (.325,.1);\draw[very thick,color=blue] (0.1625,0.1) circle (2pt);}}
\newcommand{\synchtraj}{\tikz{\node [invdot] at (0,0){};\draw [ultra thick, dashed,color=gray, draw opacity=0.7] (0,0.1) -- (.325,.1);}}
\newcommand{\tp}{\textnormal{\textsf{T}}}
\newcommand{\set}[1]{\mathcal{#1} }
\renewcommand{\vec}[1]{\mbox{\boldmath{$#1$}}}
\newcommand{\mat}[1]{\mbox{\boldmath{$#1$}}}

\DeclareMathAlphabet\mathbb{U}{fplmbb}{m}{n}
\newcommand{\vecgrk}[1]{\mbox{\boldmath{$#1$}}}

\newcommand{\rank}[1]{\text{rank}\left(#1\right)}
\newcommand{\card}[1]{\text{card}\left(#1\right)}

\newcommand{\trace}[1]{\text{trace}\left(#1\right)}
\newcommand{\circledtext}[1]{{\large\textcircled{\normalsize \texttt{#1}}}}

\DeclareMathAlphabet\mathbb{U}{fplmbb}{m}{n}
\newcommand{\RZ}{\mathbb{R}}

\renewcommand{\Re}[1]{\mathfrak{Re}\{#1\}}
\renewcommand{\Im}[1]{\mathfrak{Im}\{#1\}}

\def \Gg {\set{G}}                       % graph G
\def \Vg {\set{V}_{\Gg}}                       % vertex set
\def \Eg {\set{E}_{\Gg}}                       % edge set
\def \Lg {\mat{L}_\Gg}                    % Laplacian of a graph
\def \aji {a_{\Gg_{ji}}}                   % adjacency matrix entries
\def \aij {a_{\Gg_{ij}}}
\def \lij {l_{\Gg_{ij}}}		% Laplacian matrix entries
\def \sumjN {\sum_{j=1}^N}              % denoting the sum used in agent dynamics
\def \zero {\vec{0}}                    % zero vector
\def \tauvec {\vecgrk{\tau}}
\def \x {\vec{x}}                       % state vector agent
\def \xol {\overline{\x}}               % state vector exosystem
\def \u {\vec{u}}                       % input vector agent
\def \uol {\overline{\u}}               % input vector exosystem        
\def \y {\vec{y}}                       % output vector agent
\def \yol {\overline{\y}}               % output vector exosystem
\def \e {\vec{e}}						% unit vector
\def \yt {\widetilde{\vec{y}}}				% stationary tracking error
\def \mPi {\mat{\Pi}}                   % solution regulator eq.
\def \mGamma {\mat{\Gamma}}             % solution regulator eq.
\def \mTheta {\mat{\Theta}} 			% system matrix of the hamiltonian system
\def \N {\mat{N}}
\def \M {\mat{M}}
\def \Q {\mat{Q}}                       % tracking error weight
\def \R {\mat{R}}                       % input energy weight
\def \G {\mat{G}}                       % stationary integrand
\def \E {\mat{E}}                       % matrix ICCA Paper
\def \Z {\mat{Z}}                       % slack variable
\def \X {\mat{X}}                       % slack variable
\def \A {\mat{A}}                       % system matrix agent
\def \Aol {\overline{\A}}               % system matrix exosystem
\def \B {\mat{B}}                       % input matrix agent
\def \Bol {\overline{\B}}               % input matrix exosystem
\def \C {\mat{C}}                       % output matrix agent
\def \Col {\overline{\C}}               % output matrix exosystem
\def \nol {\overline{n}}                % state dimension exosystem
\def \mol {\overline{m}}                % input dimension exosystem
\def \K {\mat{K}}                       % feedback matrix for agent
\def \Kol {\overline{\K}}               % synchronization gain matrix for exosystem
\def \P {\mat{P}}                       % Lyapunov matrix
\def \Pol {\overline{\P}}                       
\def \I {\mat{I}}                       % identity matrix
\def \Ool {\overline{\Omega}}           % ferquency spectrum of system matrix exosystem

\newcommand{ \ma}{\begin {bmatrix}}
\newcommand{\me }{ \end {bmatrix}}
\newcommand{ \mra}{\begin {pmatrix}}
\newcommand{\mre }{ \end {pmatrix}}

\newcommand{\figref}[1]{Fig.~\ref{#1}}
\newcommand{\thmref}[1]{Theorem~\ref{#1}}
\newcommand{\optref}[1]{OP.~\ref{#1}}
\newcommand{\secref}[1]{Section~\ref{#1}}
\newcommand{\tabref}[1]{Table~\ref{#1}}
\newcommand{\defref}[1]{Definition~\ref{#1}}
\newcommand{\defrefv}[2]{\defref{#1}.\ref{#2})}

\newcommand{\assref}[1]{Asmp.~\ref{#1}}
\newcommand{\remref}[1]{Remark~\ref{#1}}
\newcommand{\lemref}[1]{Lemma~\ref{#1}}

\newcommand{\algoref}[1]{Algorithm~\ref{#1}}

\newtheorem{lemma}{Lemma}
\newtheorem{theorem}{Theorem}
\newtheorem{assumption}{Assumption}

\newtheorem{defn}{Definition}
\newtheorem{remark}{Remark}

\newtheorem{optprob}{Optimization Problem}

\title{\LARGE \bf
Optimal Stationary Synchronization of Heterogeneous\\ Linear Multi-Agent Systems
}

\author{Sebastian Bernhard, Saman Khodaverdian and J\"urgen Adamy% <-this % stops a space
\thanks{The authors are with Institute of Automatic Control and Mechatronics; Control Methods and Robotics Lab,
        Technische Universit\"at Darmstadt, Landgraf-Georg Str. 4, 64283 Darmstadt, Germany,
        {\{\tt\small bernhard, saman.khodaverdian, adamy\}}@rmr.tu-darmstadt.de}%
}

\begin{document}

\maketitle
\thispagestyle{empty}
\pagestyle{empty}

%%%%%%%%%%%%%%%%%%%%%%%%%%%%%%%%%%%%%%%%%%%%%%%%%%%%%%%%%%%%%%%%%%%%%%%%%%%%%%%%
\begin{abstract}
In this paper, we address the output synchronization of heterogeneous linear networks. In the literature, all agents are typically required to synchronize exactly to a common trajectory. Here, we introduce optimal stationary synchronization (OSS) instead which permits non-zero steady-state synchronization errors. As a benefit, we are able to relax standard requirements. E.g., agents are allowed to participate in the network even when they usually cannot synchronize exactly. In addition, OSS enables agents to save input-energy by synchronizing within tolerable error-bounds. Our new method combines the synchronization of bounded exosystems with local infinite-time linear quadratic tracking (LQT). This results in an optimal balance of each agent's synchronization error versus its consumed input-energy. Moreover, we extend recent results in LQT such that the derived time-invariant optimal control guarantees that the synchronization error satisfies given strict bounds. All these aspects are demonstrated by an illustrative simulation example with a detailed analysis.
\end{abstract}

%%%%%%%%%%%%%%%%%%%%%%%%%%%%%%%%%%%%%%%%%%%%%%%%%%%%%%%%%%%%%%%%%%%%%%%%%%%%%%%%
\section{Introduction}
This paper considers the output synchronization problem for linear heterogeneous \emph{multi-agent systems} (MAS). MAS play an important role in various research areas \cite{Olfati-SaberFaxMurray2007, RenBeardAtkins2007}.

An internal model principle has proven to be necessary and sufficient for synchronization \cite{WielandSepulchreAllgoewer2011}. Loosely speaking, some part of the agents' dynamics has to be identical, which is not satisfied for heterogeneous agents  in general.
One way for solving this problem is to homogenize the agents by local feedback and then to achieve synchronization with the help of classical methods, cf. e.g. \cite{Khodaverdian2014-ifacWC, Khodaverdian2015-cdc, YangSaberiStoorvogelGrip2014}.
However, such approaches are limited in their applicability.
Alternatively, it is possible to include identical virtual exosystems into a dynamic control strategy.
These define a mutual objective of all agents. Then the homogeneous exosystems are synchronized to the synchronization trajectory $\overline{\y}(t)$.
Hence, the problem of \textit{exact synchronization} (EXS) reduces to a local trajectory tracking task, i.e. the synchronization error has to vanish: $\lim_{t \to \infty}\yt_i(t)=\lim_{t \to \infty}\big(\y_i(t) - \yol(t)\big)=\zero$.
First results of this approach were carried out by \cite{KimShimSeo2011}, \cite{WielandSepulchreAllgoewer2011}. 

In this paper, we consider the question: Is such an exact synchronization (EXS) always meaningful or necessary in heterogeneous multi-agent systems?

We believe the answer is: No. Especially for heterogeneous networks, the requirement of EXS can be quite restrictive. E.g., suppose that an agent is incapable of achieving a desired objective. If EXS is forced then all agents will have to synchronize to a common trajectory necessarily differing from the objective, cf. \cite{WielandSepulchreAllgoewer2011}. 
Motivated by biological considerations, it sounds more natural to us that such ``weak'' individuals try to follow the objective as best as they can instead of dictating all other agents to fail to do so. Moreover, it is easy to think of situations when agents have to consider additional requirements. E.g., saving energy in order to be able to participate in the network for a given time period. From this practical point of view, a synchronization within defined acceptable bounds seems more reasonable. Then the key question is: How can this new degree of freedom be used for an optimized performance of each agent without increasing the complexity of the control structure?

In the literature, however, little has been done so far and many results are similar to EXS at heart. E.g., an $\mathcal H_\infty$-Norm ``almost synchronization'' is presented in \cite{Peymani2014} where EXS is assumed in absence of disturbances. Or, ``practical synchronization'' is introduced in \cite{Montenbruck2015} which requires that arbitrarily small bounds on synchronization errors are implementable. The same is true for ``funnel synchronization'' \cite{Shim2015}. Altogether, it is not covered how weakening the requirement of EXS can be exploited to the benefit of the agents.

In this context, we propose a linear-quadratic tracking (LQT) approach for \textit{optimal stationary synchronization} (OSS) of heterogeneous agents. It relies on recent results in infinite-time LQT \cite{Bernhard2017b}. For the first time, to the best of our knowledge, we will present a local, time-invariant optimal control with respect to quadratic cost
\begin{equation} \label{eq_cost}
\hspace{-.25cm}J_{t_\text{f}}\big(\yt_{i}(\cdot),\u_{i}(\cdot)\big)\hspace{-.075cm}=\hspace{-.075cm}\tfrac{1}{2}\hspace{-.15cm}\int_0^{t_\text{f}}\hspace{-.2cm} \yt_{i}(t)^\tp \Q_i \yt_{i}(t)+\u_i^\tp(t)\R_i \u_i(t) \hspace{-.075cm}\text{ d}t\hspace{-.15cm}
\end{equation}
on infinite horizons $t_{\text{f}}\rightarrow\infty$, for which $\lim_{t\to\infty}\yt_{i}(t)\neq\zero$. We suppose the weights $\Q_i\succ0$ and $\R_i\succ0$ are additional design parameters. These will allow each agent to balance the importance of synchronization versus input-energy consumption individually -- even when EXS is infeasible, e.g. due to under-actuation for less inputs than outputs.

Notice that finding an optimal control is not a trivial task since $\lim_{t_{\text{f}}\to\infty}J_{t_{\text{f}}}(\cdot)=\infty$ for any $\u_{i}(\cdot)$ in general~\cite{Anderson2007}. Nevertheless, under reasonable assumptions on infinite horizons, i.e. bounded $\yol(t)$, \cite{Bernhard2017b} derives a time-invariant control which is proven to be strongly optimal considering an equivalent LQT problem. This forms the basis of our approach. We will carry out some modifications to adapt the results to MAS, e.g. a definition of stationary optimality at the end of \secref{sec_mas}. Then we are ready to achieve OSS in \secref{sec_opt_syn}. Exploiting results in \cite{Bernhard2016}, we are also able to introduce a parametric optimization problem (OP) whose solution satisfies the algebraic equations in~\cite{Bernhard2017b}.

As discussed above, a certain bound on the $j$-th component of the synchronization error: $\vert \widetilde y_{ij} \vert \leq \epsilon_{ij}$ is often desired. To this end, we will introduce an OP in \secref{sec_opt_err_syn}  which constitutes an inverse problem in a wider sense. Meaning that the goal is to obtain a $\Q_i$ which leads to an input-energy efficient optimal control so that given feasible bounds are satisfied. We call this an \textit{error-bounded} OSS (EBOSS). Here, the objective function will be motivated by the OP previously mentioned. The OP in question involves bilinear and linear  matrix (in)equalities (BMI, LMI); hence, an efficient path-following algorithm, e.g. see \cite{Ostertag2008}, is implemented.
 
Summarizing, our \textbf{novel contribution} is: Based on a dynamic control strategy, the synchronization of the agents' identical exosystems gives a desired common synchronization trajectory $\yol(t)$. Then, considering each agent's synchronization error $\yt_{i}(t)=\y_i(t)-\yol(t)$, we derive a local time-invariant control $\u_i^*(\cdot)$ from algebraic equations or parametric optimization, which
\begin{enumerate}
	\item[C1)] leads to optimal stationary synchronization with respect to cost \eqref{eq_cost}, $t_{\text{f}}\rightarrow\infty$ for any initial conditions of the agents' dynamics and exodynamics (OSS)
	\item[C2)] and can be obtained for quadratic, over- and under-actuated agents as well as under relaxed assumptions.
	\item[C3)] guarantees error-bounded OSS, i.e. given error-bounds $\vert \widetilde y_{ij} \vert \leq \epsilon_{ij}$, $\forall j$ are additionally satisfied for all relevant initial values of the agents' exosystems. (EBOSS)
\end{enumerate}
The paper is structured as follows: First, the framework of MAS along with assumptions and an optimality definition are presented in \secref{sec_mas}. Second, C1-2) and C3) are derived on a local level in \secref{sec_opt_syn} and \ref{sec_opt_err_syn}, respectively. This underlines that our results can be generalized for tracking tasks involving exosystems. Before our final conclusions, simulation results in \secref{sec_sim} account for C1-3).

\textit{Mathematical notations:}
The zero and identity matrix have appropriate dimensions if not stated explicitly: $\zero_{a\times b}$ or $\I_a$. A matrix $\vec{M}$ is positive (semi-)definite if $\vec{M}\succ(\succeq)\zero$.
The number of unique elements of a multiset $\Omega$ is given by $\card{\text{supp}(\Omega)}$ and for an element $k \in \Omega$ the multiplicity is $m_\Omega(k)$. The unit vector $\e_i$ of appropriate length has $i$-th element equal to one, zero else. $\mathbb{Q}$ denotes the set of rational numbers. The convex hull of a set of vectors $\set{X}$ is $\text{conv}(\set{X})$. By $\text{diag}(\A,\B,\ldots)$, we define a block-diagonal matrix.
\section{HETEROGENEOUS LINEAR MULTI-AGENT-SYSTEMS}
\label{sec_mas}

In this section, we present the structure of the MAS and give the agents' dynamics and necessary assumptions. Then, the synchronization gain for the homogeneous exosystems is determined. Finally, we introduce the important definition of optimal stationary synchronization (OSS).

Furthermore, we have to give technical requirements for the structure of the exosystem in \secref{sec_problem} and for the set of initial values of the exosystems in \secref{sec_sync}. Since these are not necessary to understand the main results they may be skipped at first. For understanding of the proofs and for implementation,  they should be closely followed. The context should be clearer after studying \secref{sec_sim}.

%%%%%%%%%%%%%%%%%%%%%%%%%%%%%%%%%%%%%%%%%%%%%%%%%%%%%%%%%%%%%%%%%%%%%%%%%%%%%%%%%%%%%%%%%%%%%%%%%%%%%%
%%%%%%%%%%%%%%%%%%%%%%%%%%%%%%%%%%%%%%%%%%%%%%%%%%%%%%%%%%%%%%%%%%%%%%%%%%%%%%%%%%%%%%%%%%%%%%%%%%%%%%
\subsection{System Setup}
%%%%%%%%%%%%%%%%%%%%%%%%%%%%%%%%%%%%%%%%%%%%%%%%%%%%%%%%%%%%%%%%%%%%%%%%%%%%%%%%%%%%%%%%%%%%%%%%%%%%%%
\subsubsection{Graph Theory}
\label{sec_graph}

We model the information exchange in the multi-agent system by a time-invariant \emph{directed graph} $\Gg = (\Vg,\Eg)$. The $i$-th agent in the network is represented by vertex $i \in \Vg = \{1,\ldots,N\}$, $N < \infty$. Agent $j$ receives information from agent $i$ if the edge $(i,j) \in \Eg$ exists. 
A \emph{Laplacian matrix} describes the communication network \cite{Olfati-SaberFaxMurray2007} and is
defined as $\Lg = [\lij] \in \RZ^{N \times N}$ with
\begin{equation*}
 \lij = \begin{cases}
	 \sum_{k=1}^N a_{\Gg_{ki}}, & i=j, \\ -\aji, & i \ne j, 
       \end{cases}
 \quad 
 \aij = \begin{cases}
	 1, & (i,j) \in \Eg, \\ 0, & (i,j) \not \in \Eg.
       \end{cases}
\end{equation*}
 
\begin{defn} \label{def_connectedness}
 A directed graph $\Gg = (\Vg,\Eg)$ contains a \emph{directed spanning tree} if there exists at least one vertex that can reach every other vertex, using the edges contained in the set $\Eg$. 
\end{defn}

It can be shown that a directed spanning tree exists if and only if $\Lg$ has a simple eigenvalue in zero \cite{LafferriereWilliamsCaughmanVeerman2005}, 
i.e. $\lambda_1(\Lg)=0$ and $\lambda_i(\Lg) \ne 0$ for $i \in \{2,\ldots,N\}$.

%%%%%%%%%%%%%%%%%%%%%%%%%%%%%%%%%%%%%%%%%%%%%%%%%%%%%%%%%%%%%%%%%%%%%%%%%%%%%%%%%%%%%%%%%%%%%%%%%%%%%%
\subsubsection{Agent Dynamics \& Assumptions}
\label{sec_problem}

We consider a heterogeneous network of $N$ agents, with $i$-th agent 
\allowdisplaybreaks
\begin{subequations} \label{eq_agent_dyn}
\begin{align} 
   \dot{\x}_i &= \A_i \x_i + \B_i \u_{i} ,  \label{eq_het_dyn}\\   
         \y_i &= \C_i \x_i ,  \\
         \u_i &= -\K_i (\x_i - \mPi_i \xol_i) + \mGamma_i \xol_i ,  \label{eq_u-track}  \\
 \dot{\xol}_i &= \Aol \xol_i + \Bol \uol_i ,	\label{eq_hom_sys} \\
 \yol_i &= \Col \xol_i ,  \\
       \uol_i &= - \Kol \sumjN \aji (\xol_i - \xol_j)  \label{eq_u-sync} 
\end{align}
\label{eq_sshetMAS}%
\end{subequations}
with state, input and output vector $\x_i \in \RZ^{n_i}$, $\u_i \in \RZ^{m_i}$ and $\y_i \in \RZ^{p}$. 
Since the network is heterogeneous, the system, input and output matrices: $\A_i \in \RZ^{n_i \times n_i}$, $\B_i \in \RZ^{n_i \times m_i}$ and 
$\C_i \in \RZ^{p \times n_i}$ can be different among the agents, with possibly different state and input dimensions, but the output dimension must be equal.
The $i$-th dynamic control strategy is given by (\ref{eq_agent_dyn}\text{c-f}). Herein, $\xol_i \in \RZ^{\nol}$ are states of an exosystem and $\uol_i \in \RZ^{\mol}$ its input. 
The exosystems determine a task which the network should accomplish and, hence, are homogeneous. It is defined by identical matrices $\Aol \in \RZ^{\nol \times \nol}$, $\Col \in \RZ^{p \times \nol}$ and
a $\Bol \in \RZ^{\nol \times \mol}$ such that $(\Aol,\Bol)$ is stabilizable.  
The matrices $\K_i \in \RZ^{m_i \times n_i}$, $\mPi_i \in \RZ^{n_i \times \nol}$, $\mGamma_i \in \RZ^{m_i \times \nol}$ and $\Kol \in \RZ^{\mol \times \nol}$ are to be designed.

The following assumptions are made for all agents:
\assumption $\Gg$ contains a directed spanning tree. \label{assump_spanning_tree}
\assumption $(\A_i,\B_i,\C_i)$ is stabilizable and detect- able. \label{assump_ctrb_obsv}
\assumption All eigenvalues $\lambda_j(\Aol)$ have equal algebra- ic and geometric multiplicities and satisfy $\Re{\lambda_j(\Aol)}=0$.\label{assump_eig_Aol}%

Let us define the multiset $\Ool = \{\Im{\lambda_j}\geq0 \text{ }\vert\text{ } \lambda_j \in\sigma(\Aol), \forall j\}$ which we will call the frequency spectrum.
\assumption  It holds $\frac{\omega_i}{\omega_j}\in \mathbb{Q}$ $\forall \omega_i, \omega_j\neq0 \in \Ool$.\label{assump_freq_ratio}

\assref{assump_spanning_tree} is a necessary condition to achieve synchronization with distributed synchronization protocols in time-invariant networks, and \assref{assump_ctrb_obsv} is standard in control theory.
\assref{assump_eig_Aol} guarantees bounded references given by the exosystem which is a standard assumption in context of infinite-time optimal tracking. Moreover, we regard periodic synchronization trajectories here. Since $\mathbb{Q}\subset\mathbb{R}$ is dense, however, \assref{assump_freq_ratio} is not a restriction effectively. Then a time period $T\in\mathbb{R}$ of the exosystem exists such that $\forall \omega_j\neq0 \in \Ool$ $\exists k_j\in \mathbb{N}$ such that $T=k_j \frac{2\pi}{\omega_j}$ holds. We remark that we do not need to calculate $T$ to apply the results of this paper.	

Furthermore, we assume without loss of generality that the system matrix of the exosystem is organized as follows
\begin{equation}
	\Aol = \text{diag}\left(\Aol_0,\Aol_1,\ldots,\Aol_{N_{\Ool}}\right) \label{eq_Aol}
\end{equation}
with the number of different circular frequencies $N_{\Ool}=\card{\text{supp}\left(\Ool\right)}-1$, where we assumed that $0\in\Ool$, and
\begin{equation*}\begin{split}
\Aol_0 &= \zero_{m_{\Ool}(0)\times m_{\Ool}(0)},\\
\Aol_j &= \omega_j\left(\I_{m_{\Ool}(\omega_j)}\otimes\ma 0 & 1\\ -1 & 0 \me\right)
\end{split}
\end{equation*}
for $j=1,\ldots,N_{\Ool}$, $\Ool\ni\omega_j\neq0$ and $\omega_i\neq\omega_j$ unless $i=j$. 
With respect to $\Aol_0$, we define the constant scalar state $\overline x_l$, $l\in\{1,\ldots,L\}$  with $L= m_{\Ool}(0)$. Furthermore, we define the state $\hat{\xol}_h \in\RZ^2$ of each harmonic second-order subsystem, $h\in\{1,\ldots,H\}$ with $H=\sum_{j=1}^{N_{\Ool}}m_{\Ool}(\omega_j)$.

 At this point, the block-diagonal structure of $\Aol$, which can always be obtained by similarity transformation, may seem technical. However, it will permit us to make use of some helpful results of \cite{Bernhard2016}.

\remark Without loss of generality, we disregarded heterogeneous disturbances in (\ref{eq_sshetMAS}a-b). Based on \cite{Bernhard2017b}, all presented results can be extended to disturbances given by local autonomous systems as long as \assref{assump_eig_Aol} and \ref{assump_freq_ratio} hold.

\remark In view of contribution C2), typical assumptions such as $\rank{\B_i}\geq\rank{\C_i}$ and that the eigenvalues of $\Aol$ and the invariant zeros of the agents' dynamics are disjoint, e.g. both is assumed in \cite{KimShimSeo2011}, are not yet required. These are usually needed to guarantee the feasibility of EXS. In contrast, the assumptions can be weakened for OSS in \secref{sec_opt_syn}. E.g., agents with less inputs than outputs are feasible, cf. the example in \secref{sec_sim}. \label{rem_asmp}

%%%%%%%%%%%%%%%%%%%%%%%%%%%%%%%%%%%%%%%%%%%%%%%%%%%%%%%%%%%%%%%%%%%%%%%%%%%%%%%%%%%%%%%%%%%%%%%%%%%%%%
\subsection{Synchronization of Exogenous Systems}
\label{sec_sync}

%%%%%%%%%%%%%%%%%%%%%%%%%%%%%%%%%%%%%%%%%%%%%%%%%%%%%%%%%%%%%%%%%%%%%%%%%%%%%%%%%%%%%%%%%%%%%%%%%%%%%%
Synchronization of exosystem states, i.e. $\lim_{t \to \infty} \big( \xol_i(t) - \xol_j(t) \big) = \zero$ for all $i, j \in \{1,\ldots,N\}$, with the distributed control law \eqref{eq_u-sync} occurs if and only if $\Aol - \lambda_i(\Lg) \Bol \Kol$ is Hurwitz for all $i \in \{2,\ldots,N\}$, e.g. \cite{MaZhang2010}. 
The following lemma is taken from \cite{Tuna2008} 
and given without proof.
\begin{lemma} \label{lem_homsync}
 Let $(\Aol, \Bol)$ be stabilizable and the symmetric matrix $\Pol$ be the unique positive definite solution of the \textit{algebraic Riccati equation}
 \begin{equation*}
  \Aol^\tp \Pol + \Pol \Aol - \Pol \Bol \Bol^\tp \Pol + \I_n = \zero.
 \end{equation*}
 The matrix $\Aol - \lambda_i(\Lg) \Bol \Kol$ is Hurwitz for all $i \in \{2, \ldots, N\}$, if the synchronization gain is chosen as $\Kol = \sigma^{-1} \Bol^\tp \Pol$ with $0 < \sigma \le \min_{i\ge2}\{\Re{\lambda_i(\Lg)}\}$.  \hfill\QED
\end{lemma}

%%%%%%%%%%%%%%%%%%%%%%%%%%%%%%%%%%%%%%%%%%%%%%%%%%%%%%%%%%%%%%%%%%%%%%%%%%%%%%%%%%%%%%%%%%%%%%%%%%%%%%
Following \cite{Tuna2008}, all outputs of the agents' exosystem converge to the synchronization trajectory $\yol(t) = \Col \xol(t)$ with%
\begin{subequations}\begin{align} 
	\xol(t) &= \text{e}^{\text{\small$\Aol$} t} \xol(0) , \label{eq_xol_t} \\
	\xol(0) &\in \text{conv}\big( \{\xol_1(0), \ldots, \xol_N(0) \}\big) . \label{eq_xol_0}  
	\end{align} \label{eq_syn_traj}
\end{subequations}
Since (\ref{eq_sshetMAS}d-e) defines a mutual objective, it is reasonable to assume that $\xol_i(0) \in \set{\overline X}$, $\forall i$ where $\set{\overline X}$ is a bounded subset of the euclidean space: $\set{\overline X}\subset\RZ^{\nol}$.
Due to \eqref{eq_xol_0}, it results $\xol(0) \in \set{\overline X}$. With respect to the structure of \eqref{eq_Aol}, we suppose that $\set{\overline X}$ accounts for the maximal step-height $a_l^{\text{max}}$ of each scalar constant subsystem $\overline x_l$, $l\in\{1,\ldots,L\}$ and for the maximal amplitude $\widehat A_h^{\text{max}}$ of each harmonic second-order subsystem $\hat{\xol}_h$, $h\in\{1,\ldots,H\}$, cf. \secref{sec_problem}. This means that any $\xol(0)$ with $\vert\overline x_l(0)\vert\leq a_l^{\text{max}}$, $\forall l$ and $\Vert\hat{\xol}_h(0)\Vert_2\leq\widehat A_h^{\text{max}}$, $\forall h$ satisfies $\xol(0)\in\set{\overline X}$.

Hence, we may write $\set{\overline X}=\left(\cap_{l=1}^L \set{\overline X}_l\right)\cap\left(\cap_{h=1}^H \set{\overline X}_h\right)$ with%
\begin{equation}\begin{split}
	\set{\overline X}_l &= \bigg\{\xol \in \mathbb{R}^{\nol} \,\Big\vert\, \tfrac{1}{\left(a_l^{\text{max}}\right)^2}\xol^\tp\M_l\xol\leq 1\bigg\}, \notag \\
	\set{\overline X}_h &= \bigg\{\xol \in \mathbb{R}^{\nol} \,\Big\vert\, \tfrac{1}{\left(\widehat A_h^{\text{max}}\right)^2}\xol^\tp\N_h\xol\leq 1\bigg\}\notag
\end{split}\end{equation}
where
\begin{subequations}\label{eq_M_N}\begin{align}
	\M_l &= \text{diag}\left(\vec{e}_l\vec{e}_l^\tp,\zero_{2H\times 2H}\right), \notag \\
	\N_h &= \text{diag}\left(\zero_{L\times L},(\vec{e}_h\vec{e}_h^\tp\otimes\I_2)\right) \notag
	\end{align}\end{subequations}
are diagonal matrices with $\e_l\in\mathbb{R}^L$ and $\e_h\in\mathbb{R}^H$. 
It is important to note that $\set{\overline X}$ is an \textit{invariant set} implying that the synchronization trajectory \eqref{eq_xol_t} satisfies $\xol(t)\in\set{\overline X}$ $\forall t\in[0,\infty)$ if $\xol(0) \in \set{\overline X}$. Furthermore, let us define
\begin{equation}
	\P=\text{diag}\left(a_1^{\text{max}},\ldots,a_L^{\text{max}},\widehat A_1^{\text{max}}\I_2,\ldots,\widehat A_H^{\text{max}}\I_2\right) \label{eq_P}
\end{equation}
which will be used for a normalization later on.

The preceding definitions are used in the optimization problem formulated in \secref{sec_opt_err_syn}. We remark that the convex set $\set{\overline X}$ may be given in a different form than above, e.g. by a convex polytope. However, then it may be necessary to approximate $\set{\overline X}$ by an invariant set based on quadratic forms as in \cite[Sec. 2.6.3]{Boyd1994} in order to apply the results in \secref{sec_opt_err_syn}.

%%%%%%%%%%%%%%%%%%%%%%%%%%%%%%%%%%%%%%%%%%%%%%%%%%%%%%%%%%%%%%%%%%%%%%%%%%%%%%%%%%%%%%%%%%%%%%%%%%%%%%
%%%%%%%%%%%%%%%%%%%%%%%%%%%%%%%%%%%%%%%%%%%%%%%%%%%%%%%%%%%%%%%%%%%%%%%%%%%%%%%%%%%%%%%%%%%%%%%%%%%%%%
\subsection{Local Transition \& Definition of OSS}

Once the homogeneous part of the agent-dynamics, i.e. the exosystems, are synchronized, the problem of synchronizing $\y_i$ $\forall i \in \{1,\ldots,N\}$ has to be solved locally. Hence, the task of each agent is that its output $\y_i$ tracks the output $\yol_i$ of its exosystem stationarily to some specified degree.

For this reason, we introduce the pair $(\mPi_i,\mGamma_i)$ satisfying
\begin{equation} \label{eq_stat_sys}
\mPi_i \Aol = \A_i \mPi_i +\B_i \mGamma_i.
\end{equation}
If $\K_i$ is chosen such that $\A_i - \B_i \K_i$ is Hurwitz, the local transition $\lim_{t \to \infty}\big(\x_i(t)-\mPi_i \xol_i(t) \big)=\zero$ will be guaranteed. Omitting details, this results by standard means \cite{Trentel2001} since $\lim_{t \to \infty} \uol_i(t) = \zero$ implies that \eqref{eq_hom_sys} is asymptotically autonomous. In addition, $\lim_{t \to \infty} \big( \xol(t) - \xol_i(t) \big) = \zero$ holds; hence, it follows $\lim_{t \to \infty}\big(\x_i(t)-\mPi_i \xol(t) \big)=\zero$. As a consequence, for analyzing each agent's stationary behavior based on \eqref{eq_het_dyn} and \eqref{eq_u-track} it suffices to analyze its stationary response $\mPi_i\xol(t)$ due to excitation by $\mGamma_i\xol(t)$ with \eqref{eq_syn_traj}.

The main goal of this contribution is to guarantee an optimal stationary synchronization (OSS) by a distributed control. Since $J(\yt_{i},\u_{i})\rightarrow\infty$, $t_\text{f}\rightarrow\infty$ for any $\u_{i}(\cdot)$ in general, the classical definition of optimality does not apply here \cite{Anderson2007}. For the sake of compactness, we avoid to introduce technical concepts of optimality for infinite-time LQT. However, it can be drawn from \cite{Bernhard2017b} that a solution satisfying the following definition of OSS is a so-called \textit{strongly optimal} solution of an equivalent LQT problem.

\begin{defn} With respect to the cost \eqref{eq_cost} and any $\xol(0) \in \RZ^{\nol}$, the \textbf{stationary synchronization} of agent $i$ for the local control $\u_{i}(\cdot)$ given by \eqref{eq_u-track} is
	\begin{enumerate}
		\item \textbf{exact} if $(\mPi_i,\mGamma_i)$ such that $\lim_{t \to \infty}\yt_{i}(t)=\zero$. (EXS)\label{def_exact_syn}
		\item \textbf{optimal} if $(\mPi_i^*,\mGamma_i^*)$ such that for any other $\widehat\u_{i}(\cdot)$ 
		\begin{equation*}
			\lim_{t_\text{f} \to \infty}\left(J_{t_\text{f}}\big(\widehat\yt_{i},\widehat\u_{i}\big)-J_{t_\text{f}}\big(\yt_{i}^*,\u^*_{i}\big)\right)=+\infty
		\end{equation*}
		holds if $\widehat\x_{i}(t)-\mPi_i^*\xol(t)\not\rightarrow\zero$ as $t\rightarrow\infty$. (OSS)\label{def_opt_syn}
		\item \textbf{error-bounded optimal} if $(\mPi_i^*,\mGamma_i^*)$ satisfies \ref{def_opt_syn}) and, in addition, for any $\xol(0)\in\set{\overline X}$ it holds
		\begin{equation}\label{eq_err_bound}
			\vert\e_j^\tp \left(\C\mPi_i^*-\Col\right) \xol(t)\vert \leq \epsilon_{ij}%\e_j^\tp\epsvec_{i}
		\end{equation}
		with tolerated error $\epsilon_{ij}>0$, $\forall j \in \{1,\ldots,p\}$ and $\forall t \in [0,\infty)$. (EBOSS)  \label{def_opt_err_syn}		
	\end{enumerate}	
	\label{def_syn}
\end{defn}
Notice that we compare $\u_i^*(\cdot)$ to any arbitrary control $\widehat\u_i(\cdot)$. Hence, we do not impose any restrictions on the class of optimal solutions in \defrefv{def_syn}{def_opt_syn}. For exogenous references such as \eqref{eq_syn_traj}, \cite{Bernhard2017b} proves that the solution of an infinite-time LQT-problem is indeed a time-invariant control such as \eqref{eq_u-track}. This leads to an optimal stationary trajectory $\mPi^*\xol(t)$ induced by a static pre-filter $\mGamma^*\xol(t)$, i.e. the pair $(\mPi^*,\mGamma^*)$. Clearly, any other choice $(\mPi_i,\mGamma_i)$ besides $(\mPi_i^*,\mGamma_i^*)$ will require an infinite amount of additional cost based on \defrefv{def_syn}{def_opt_syn}.

\remark In \cite{Kreindler1969}, it is criticized that in infinite-time LQT there is ``no control over the resultant steady-state error''.
In contrast to \cite{Kreindler1969}, however, we will be able to explicitly consider given strict error-bounds as in \defref{def_syn}.\ref{def_opt_err_syn}) in the design process by extending the results in \cite{Bernhard2017b}.

\section{LOCAL OPTIMAL STATIONARY SYNCHRONIZATION}
\label{sec_main_opt_syn}

In this section, our contributions C1) and C3) are presented. We show how each agent achieves OSS and EBOSS by a local control $\u_i$, cf. \defref{def_opt_syn}. This allows the agent to individually balance its synchronization error in relation to its consumed input-energy. Or, the agent is enabled to synchronize as best as it can when EXS is infeasible.

To determine an optimal pair $(\mPi_i^*,\mGamma_i^*)$, OSS is addressed in \secref{sec_opt_syn} which provides useful extensions of results in \cite{Bernhard2017b}. These will help us to approach the EBOSS in \secref{sec_opt_err_syn} by means of a meaningful parametric optimization problem with optimization variable $\Q_i$.

The results presented here account for infinite-time LQT-problems in general. Hence, we drop the index $i$ in the sequel to emphasize the modularity of our approach.

%%%%%%%%%%%%%%%%%%%%%%%%%%%%%%%%%%%%%%%%%%%%%%%%%%%%%%%%%%%%%%%%%%%%%%%%%%%%%%%%%%%%%%%%%%%%%%%%%%%%%%
\subsection{Optimal Stationary Synchronization (OSS)}
\label{sec_opt_syn}

In order to achieve optimal tracking with respect to cost \eqref{eq_cost}, we give an alternative set of equations for determining the pair $(\mPi^*,\mGamma^*)$ in comparison to \cite{Bernhard2017b}. It is given by
\begin{theorem} \label{thm_opt_syn}
	Suppose \assref{assump_ctrb_obsv} and \ref{assump_eig_Aol} are satisfied. Then, \textbf{optimal stationary synchronization} (OSS) based on \defref{def_syn}.\ref{def_opt_syn}) is achieved for any $\xol(0) \in \R^{\nol}$ if and only if $(\mPi^*,\mGamma^*)$ is given by the unique solution of the equations
	\begin{equation} \label{eq_sylv_ham}
		\ma \mPi \\ \mPi_\lambda \me \Aol = \underbrace{\ma \A & -\B\R^{-1}\B^\tp \\ -\C^\tp\Q\C & -\A^\tp  \me}_{\text{\small=\mTheta}} \ma \mPi \\ \mPi_\lambda \me + \ma \zero \\ \C^\tp \Q \Col \me\vspace{-.2cm}
	\end{equation}
	with $\mPi_\lambda \in \RZ^{n\times\nol}$ and
	\begin{equation}\label{eq_opt_gamma}
		\mGamma=-\R^{-1}\B^\tp\mPi_\lambda.
	\end{equation}
\end{theorem}
\begin{proof}
	%-- detectability assumption fehlt: Nein in Ann 2 enthalten
	For the present assumptions, it was proven in \cite[Thm. 10 and Corol. 11]{Bernhard2017b} that a unique static pre-filter $\mGamma^*\xol(t)$ always exists which leads to a unique stationary solution $\mPi^*\xol(t)$ satisfying \defref{def_syn}.\ref{def_opt_syn}). However, such a pair $(\mPi^*,\mGamma^*)$ must satisfy the necessary optimality conditions for infinite horizons $t_\text{f} \to \infty$ \cite{Halkin1974}. 
	Instead of the sweep-method-based approach in \cite{Bernhard2017b} (which involves an \textit{algebraic Riccati equation} -- ARE), these conditions can also be expressed by \eqref{eq_sylv_ham} and \eqref{eq_opt_gamma} in our case. Here, \eqref{eq_sylv_ham} defines the stationary solution of the \textit{Hamiltonian system}. It is well known, cf. \cite{Anderson2007}, that the corresponding system matrix $\mTheta$ does not have any eigenvalues on the imaginary axis if \assref{assump_ctrb_obsv} holds. Hence, under \assref{assump_eig_Aol}, \textit{Sylvester equation} \eqref{eq_sylv_ham} has a unique solution since $\sigma(\mTheta)\cap\sigma(\Aol)=\emptyset$, e.g. see \cite{Trentel2001}, and necessity as well as sufficiency follow by uniqueness.
\end{proof}
It will prove handy in the next section that we omitted a nonlinear ARE here. While we have already found a solution covering \defref{def_syn}.\ref{def_opt_syn}) the following optimization problem (OP) will be helpful to determine a meaningful objective for an OP accounting for \defref{def_syn}.\ref{def_opt_err_syn}). In this context, we exploit that an optimal stationary solution $\mPi^*\xol(t)$ induced by $\mGamma^*\xol(t)$ is $T$-periodic. Hence, instead of regarding the cost over $[0,\infty)$, it suffices to consider one period, i.e. $[t_0,t_0+T]$.%
\begin{lemma}\label{lem_expl_lqt}
	Under \assref{assump_ctrb_obsv}, \ref{assump_eig_Aol} and \ref{assump_freq_ratio}, the pair $(\mPi^*,\mGamma^*)$ obtained from \thmref{thm_opt_syn} is equivalently given by
	\begin{optprob} \label{opt_min_trace}
	\begin{equation*}\begin{split} %\label{eq_min_trace}
		\underset{\text{\small$\mPi,\mGamma$}}{\text{argmin }}&\hspace{.075in}\trace{\left(\C\mPi - \Col\right)^\tp\Q\left(\C\mPi - \Col\right)+ \mGamma^\tp \R \mGamma} \\
		%\trace{\ma \C\mPi - \Col \\ \mGamma \me^\tp \ma \Q & \zero \\ \zero & \R \me \ma \C\mPi - \Col \\ \mGamma \me } \\
		\text{subject to: }&\hspace{.075in}\mPi\Aol=\A\mPi+\B\mGamma.
	\end{split}\end{equation*}
	\end{optprob}
	\vspace{.05in}
\end{lemma}
\begin{proof}
	We regard the general case: there are zero and non-zero elements in $\Ool$.
	As indicated by the equality constraint, we are only interested in the stationary behavior. Hence, we examine the stationary cost with respect to \eqref{eq_cost} over one period $T$, i.e.
	\begin{equation}\label{eq_cost_T}
		\int_{t_0}^{t_0+T}\hspace{-.4cm}\xol(t)^\tp \underbrace{\left(\left(\C\mPi - \Col\right)^\tp\Q\left(\C\mPi - \Col\right)+ \mGamma^\tp \R \mGamma\right)}_{\text{\small$=\G^\tp\G$}} \xol(t) \text{ d}t
	\end{equation}
	where $t_0\geq0$ is arbitrary. Our aim is formulating a parametric OP such as \optref{opt_min_trace}. Thus, we look for a matrix $\widetilde{\G }$ such that $T \xol^\tp(t_0)\widetilde{\G }^\tp\widetilde{\G} \xol(t_0)$ equals \eqref{eq_cost_T}. With the exosystem being in the special form of \eqref{eq_Aol}, however, we can make use of \cite[Lemma 2]{Bernhard2016} which exploits the orthogonality of sinusoids. Then, \eqref{eq_cost_T} equals
	\begin{multline*}
		T \xol^\tp(t_0)\widetilde{\G }^\tp\widetilde{\G}\xol(t_0)=T\xol_0(t_0)^\tp\G_0^\tp\G_0\xol_0(t_0)\\+T\sum_{j=1}^{N_{\Ool}}\xol_j(t_0)^\tp\frac{1}{2}\left(\G_j^\tp\G_j+\E_j^\tp\G_j^\tp\G_j\E_j \right)\xol_j(t_0)
	\end{multline*}
	where $\G_j$ are the columns of $\G$ that correspond to the states $\xol_j$ associated with the $j$-th block on the diagonal of \eqref{eq_Aol} and
	\begin{equation*}
		\E_j = \I_{m_{\Ool}(\omega_j)}\otimes\ma 0 & 1\\ -1 & 0 \me.
	\end{equation*}
	Since $\E_j$ is orthogonal, we also find
	\begin{equation}\label{eq_trace_equal}
		\trace{\widetilde{\G}^\tp\widetilde{\G}}=\sum_{j=0}^{N_{\Ool}}\trace{\G_j^\tp\G_j}=\trace{\G^\tp\G}
	\end{equation}
	based on the invariance of the $\text{trace}$-operation towards similarity transformation.
		
	Based on the periodicity of any stationary solution $\mPi\xol(t)$, it is evident that an optimal pair $(\mPi^*,\mGamma^*)$ with respect to \defref{def_syn}.\ref{def_opt_syn}) must lead to a minimal cost over one period $T$. Hence, it must hold
	\begin{equation}\label{eq_LMI_cost_T}
	\hspace{-.085in}T \xol^\tp(t_0)\widetilde{\G }^\tp\widetilde{\G}\bigg\vert_{\text{\small$\mPi,\mGamma$}}\xol(t_0)\geq T \xol^\tp(t_0)\widetilde{\G }^\tp\widetilde{\G}\bigg\vert_{\text{\small$\mPi^*,\mGamma^*$}}\xol(t_0)
	\end{equation}
	for all $\xol(t_0)$ and any other $(\mPi,\mGamma)$. In the sequel, we exploit the knowledge that a unique $(\mPi^*,\mGamma^*)$ satisfying \eqref{eq_LMI_cost_T} is given by \thmref{thm_opt_syn} and show that it indeed uniquely solves \optref{opt_min_trace}.
		As a first consequence, for any other $(\mPi,\mGamma)$ we can always find an $\xol^*(t_0)$ for which the strict inequality holds in \eqref{eq_LMI_cost_T}. 
	
	Let us introduce $\Z\coloneqq\widetilde{\G }^\tp\widetilde{\G}\big\vert_{\text{\small$\mPi,\mGamma$}}-\widetilde{\G }^\tp\widetilde{\G}\big\vert_{\text{\small$\mPi^*,\mGamma^*$}}$.	
	Since $\Z\succeq \zero$, it is clear that $\trace{\Z} \geq 0$. Now suppose $\trace{\Z} = 0$ which would imply $\Z=\zero$.
	But this is a contradiction with respect to the existence of $\xol^*(t_0)$. As a result we have $\text{trace}(\Z) > 0$ and, consequently, $\trace{\G ^\tp\G\big\vert_{\text{\small$\mPi,\mGamma$}}}>\trace{\G^\tp\G\big\vert_{\text{\small$\mPi^*,\mGamma^*$}}}$ due to \eqref{eq_trace_equal}. Thus, the proposition follows. 
\end{proof}
It is well known, e.g. see \cite{Willems2004}, that the optimal infinite-time LQT-control for constant references can be obtained from an off-line OP. We have shown that this is even possible in the case of time-varying, bounded exogenous references \textit{without} a-priori knowledge of the initial value $\xol(0)$.

In a different context, an OP related to \optref{opt_min_trace} was proposed in \cite{Krener1992}. The author remarked that it is sensitive to the chosen coordinates of the exosystem, i.e. one can observe that it leads to suboptimal solutions. However, if the exosystem is transformed into the special form of \eqref{eq_Aol}, we have just proven that \optref{opt_min_trace} indeed gives the unique optimal solution.

\remark At this point, one might be tempted to solve \optref{opt_min_trace} with additional constraints \eqref{eq_err_bound}. However, this will lead to a suboptimal solution which does not account for \defref{def_syn}.\ref{def_opt_syn}). Instead, we will present a proper approach.
\remark In case \assref{assump_eig_Aol} is violated, i.e. the references are unbounded, $(\mPi,\mGamma)$ given by \thmref{thm_opt_syn} can still be applied. It constitutes an approximation of the finite-time optimal LQT-control for $t_\text{f}<\infty$ under certain conditions, for details we refer to \cite{Bernhard2017b}.

%%%%%%%%%%%%%%%%%%%%%%%%%%%%%%%%%%%%%%%%%%%%%%%%%%%%%%%%%%%%%%%%%%%%%%%%%%%%%%%%%%%%%%%%%%%%%%%%%%%%%%
\subsection{Error-Bounded Optimal Stationary Synchro. (EBOSS)}
\label{sec_opt_err_syn}

Our goal is to find a pair $(\mPi^*,\mGamma^*)$ which is optimal with respect to a cost such as \eqref{eq_cost} and satisfies the output error-bounds \eqref{eq_err_bound}, i.e. we seek an error-bounded optimal solution. In this regard, we introduce an OP which resembles an inverse problem in parts. More precisely, we look for a suitable tracking-error weight $\Q$ such that the desired bounds are satisfied by the optimal control corresponding to \defref{def_syn}.\ref{def_opt_err_syn}). At the same time, the feasible optimal control should be efficient in terms of the input-energy for a given $\R$. 
In this light, we will analyze the following
\begin{optprob}\label{opt_BMI}
	\begin{subequations}\label{eq_opt_BMI}\begin{align}
		&\min_{\text{\small$\mPi$,$\mGamma$,$\mPi_\lambda$,$\Q\succ\zero$}} \hspace{.15in}\trace{\mGamma^\tp\R\mGamma\P^2} \label{eq_opt_BMI_obj} \\ 
		\text{subject}& \text{ to:} \notag\\
		&\eqref{eq_sylv_ham}\text{, }\eqref{eq_opt_gamma}\text{ and } \forall j\in\{1,\ldots,p\}:\notag \\
		&\RZ^{L+H}\ni\tauvec_j \geq \zero\text{,} \label{eq_tau_pos}\\
		& 1-\sum_{i=1}^{L+H}\e_i^\tp\tauvec_j \geq 0\text{,} \label{eq_tau_eps}\\
		&\ma \X_j & \left(\C\mPi-\Col\right)^\tp\e_j \\ \e_j^\tp(\C\mPi-\Col) & \epsilon_j^2  \me \succeq \zero \label{eq_bound}\\
		&\text{where }\X_j=\sum_{l=1}^{L}\tfrac{\text{$\e_l^\tp\tauvec_j$}}{\left(a_l^{\text{max}}\right)^2}\M_l+\sum_{h=1}^{H}\tfrac{\text{$\e_{L+h}^\tp\tauvec_j$}}{\left(\widehat A_h^{\text{max}}\right)^2}\N_h \notag
		\end{align}\end{subequations}
\end{optprob}
\vspace{0.1in}
with $\M_l$, $\N_h$ and $\P$ as defined in \secref{sec_sync} and element-wise comparison by $\geq$.

To guarantee solvability of \optref{opt_BMI} we impose%
\assumption  There exists a pair $(\mPi,\mGamma)$ solving the regulator equations, i.e. \eqref{eq_stat_sys} and $\C\mPi-\Col=\zero$.\label{assump_EXS}

According to \cite{Trentel2001}, \assref{assump_EXS} is satisfied if and only if an EXS solution exists. Thus, arbitrarily small given $\epsilon_j>0$ can be satisfied by
$(\mPi^*,\mGamma^*)$ obtained from \thmref{thm_opt_syn} if $\Q$ is suitably chosen, i.e. a sufficiently large weighting of the synchronization error leads to a sufficiently close approximation of the EXS solution. Hence, \optref{opt_BMI} must have a solution.  Now, we are able to achieve the result:
\begin{theorem} \label{thm_opt_err_syn}
	Suppose \assref{assump_ctrb_obsv}-\ref{assump_EXS}
	hold and a set $\set{\overline X}$ is given as defined in \secref{sec_sync}. For any given $\epsilon_j >0$, $j\in\{1,\ldots,p\}$ the pair $(\mPi^*,\mGamma^*)$ obtained from the argument of \optref{opt_BMI} guarantees an \textbf{error-bounded optimal synchronization} (EBOSS) for any initial value $\xol(0)\in\set{\overline X}$ of the synchronized trajectory \eqref{eq_syn_traj}. 
\end{theorem}
\begin{proof}
	With respect to \eqref{eq_sylv_ham} and \eqref{eq_opt_gamma}, $(\mPi,\mGamma)$ is clearly constrained to satisfy the conditions in \defrefv{def_syn}{def_opt_syn}.
	
	In view of \eqref{eq_err_bound} and the invariance of $\set{\overline X}$, we only need to satisfy $\xol^\tp\left(\C\mPi^*-\Col\right)^\tp\e_j\e_j^\tp \left(\C\mPi^*-\Col\right) \xol \leq \epsilon_{j}^2$, $\forall j\in\{1,\ldots,p\}$ and any $\xol\in\set{\overline X}$. Applying the $\set{S}$-procedure as in \cite[Sec. 2.6.3]{Boyd1994}, it is sufficient if for each $j\in\{1,\ldots,p\}$ there exists $\tauvec_j\in\RZ^{L+H}$ such that \eqref{eq_tau_pos}, \eqref{eq_tau_eps} and $\X_j-\left(\C\mPi^*-\Col\right)^\tp\e_j\tfrac{1}{\epsilon_{j}^2}\e_j^\tp \left(\C\mPi^*-\Col\right) \succeq \zero $ with $\X_j$ as given above hold. By employing the \textit{Schur-Complement-Lemma}, the latter is equivalently written as \eqref{eq_bound}.	
\end{proof}
Since $(\mPi^*,\mGamma^*)$ leads to OSS due to constraints \eqref{eq_sylv_ham} and \eqref{eq_opt_gamma}, it also solves \optref{opt_min_trace} based on \lemref{lem_expl_lqt}. Comparing the objective of \optref{opt_min_trace} with \eqref{eq_opt_BMI_obj}, it is clear that \optref{opt_BMI} aims at an input-energy efficient optimal control satisfying the given bounds. By introducing a normalization by means of $\P$, which basically equals a change of coordinates by $\P^{-1}\xol$, the information on $a_{l}^\text{max}$ and $\widehat A_h^\text{max}$ is factored in the objective \eqref{eq_opt_BMI_obj}. Thus, $(\mPi^*,\mGamma^*)$ minimizes the average of the stationary input-energy $\int_{t_0}^{t_0+T}\xol(t)^\tp\mGamma^\tp \R \mGamma\xol(t) \text{ d}t$ over all $\xol(t_0)\in\set{\overline X}$ for a period $[t_0,t_0+T]$ starting at any $t_0\geq0$, cf. \cite[Sec. III-C 2)]{Hermann2018}. As a result of OSS, $(\mPi^*,\mGamma^*)$ minimizes \eqref{eq_cost_T}. Hence, any $(\mPi,\mGamma)$ which requires less stationary input-energy for a specific $\xol(t_0)$ must lead to a worse synchronization performance instead.

\remark Sufficient conditions for \assref{assump_EXS} are discussed in \remref{rem_asmp}. For under-actuated systems, \assref{assump_EXS} is typically not satisfied. Then, investigating the smallest $\epsilon_j$ feasible would be interesting which is part of future work. An easy way to use \optref{opt_BMI} still is choosing $\epsilon_j =  l_j\beta$, $0 <l_j, \beta \in \RZ$ $\forall j$ and iteratively lowering $\beta$ as long as \optref{opt_BMI} is solvable. \label{rem_eboss_it}

With respect to solving the proposed \optref{opt_BMI}, we face a difficulty. While the objective can easily be replaced by the linear objective $\min \trace {\Z\P^2}$ for slack variable $\Z \in \RZ^{\nol\times\nol}$ and constraint
\begin{equation*}
	\ma \Z & \mGamma^\tp\\ \mGamma & \R^{-1}  \me \succeq \zero,
\end{equation*}
e.g. see \cite{Boyd1994}, the constraint \eqref{eq_sylv_ham} contains a bilinear term: $-\C^\tp\Q\C\mPi$ in the variables $\Q$ and $\mPi$. Hence, \optref{opt_BMI} is effectively a BMI-problem. These types of problems are non-convex in general and particularly hard to solve which means finding a local minimum \cite{Antwerp2000}.

A way to proceed is using one of the few, freely available numerical solvers which can handle BMI. A possible choice is PENLAB \cite{Fiala2013}. However, the solver could not handle non-diagonal $\Q$, i.e. solutions denoted as ``optimal'' violated constraints. This is unfortunate since non-diagonal $\Q$ can provide better solutions in terms of a smaller objective \eqref{eq_opt_BMI_obj}.

Instead, we present an iterative method known as \textit{path-following}. We follow the basic guidelines of \cite{Ostertag2008}. The key idea is to solve a convex LMI-OP derived by first-order Taylor approximation of \optref{opt_BMI} at a current operating-point (O-P) $k-1$. An O-P is defined by a $\Q^{k-1}$ for which \optref{opt_BMI} is feasible under constraint $\Q=\Q^{k-1}$. This also gives $\mPi^{k-1}$. Then, an optimal perturbation $\Q^{k-1}+\Delta\Q^*$ is chosen by%
\begin{optprob}\label{opt_LMI}
	\begin{subequations}\label{eq_opt_LMI}\begin{align}
		&\min_{\text{\small $\mPi$,$\mGamma$,$\mPi_\lambda$,$\Delta\Q$}} \hspace{.15in}\trace{\mGamma^{\tp}\R\mGamma\P^2} \notag\\ 
		\text{sub} &\text{ject to:} \notag\\
		&\ma \mPi \\ \mPi_\lambda \me \Aol = \ma \A & -\B\R^{-1}\B^\tp \\ -\C^\tp\Q^{k-1}\C & -\A^\tp  \me \ma \mPi \\ \mPi_\lambda \me \notag \\ &\hspace{0.2in}+ \ma \zero \\ \C^\tp \Q^{k-1} \Col-\C^\tp \Delta \Q\left(\C\mPi^{k-1}-\Col\right) \me\text{,} \notag \\
		&\eqref{eq_opt_gamma}\text{, }(\ref{eq_opt_BMI}\text{b-d})\text{,} \notag\\
		&\ma \alpha^{k-1}\Q^{k-1} & \Delta\Q \\ \Delta\Q & \alpha^{k-1}\Q^{k-1} \me\succ\zero\text{,} \label{eq_LMI_alpha_Q}\\
		& \Q^{k-1}+\Delta \Q \succ \zero. \notag
		\end{align}\end{subequations}
\end{optprob}
\vspace{0.1in}
As suggested by \cite{Ostertag2008}, the constraint \eqref{eq_LMI_alpha_Q} guarantees $\alpha^{k-1}\Vert\Q^{k-1}\Vert_2 > \Vert\Delta\Q^*\Vert_2$ with $\alpha^{k-1}>0$ which permits only a local search around the O-P. 

Next, we have to check if \optref{opt_BMI} under additional constraint $\Q^{k}=\Q^{k-1}+\Delta \Q^*$ is feasible which would yield $\mPi^k$, $\mGamma^k$.
Suppose this is true and, in addition, a relative decrease of the objective $\Delta_\text{rel}^k\hspace{-.25mm}=\hspace{-.25mm}1-\;\nicefrac{\trace{\text{$\mGamma^{k\,\tp}\R\mGamma^{k}\P^2$}}}{\trace{\text{$\mGamma^{k-1\,\tp}\R\mGamma^{k-1}\P^2$}}}\hspace{-.25mm}>\hspace{-.25mm}0$ took place. Only then, the new O-P given by $\Q^{k}=\Q^{k-1}+\Delta \Q^*$, $\mPi^k$ is accepted. Otherwise $\Delta \Q^*$ is discarded, i.e. $\Q^{k}=\Q^{k-1}$, $\mPi^k=\mPi^{k-1}$.

Before the next iteration is executed, an adaptation of $\alpha^{k}$ is performed \cite{Ostertag2008}. Due to similarities to trust-region algorithms, a typical adaptation with case analysis can look like
\begin{subnumcases}{\hspace{-0.5cm}\alpha^{k}=}
	 \min{\left(\gamma\alpha^{k-1},\alpha_{\text{max}}\right)}\text{,}\hspace{-0.35cm} & if new O-P accepted \label{eq_alpha_true}\\
	 \delta\alpha^{k-1}\text{,} & if new O-P rejected \label{eq_alpha_false}
\end{subnumcases}
with $\alpha_\text{max}>0$, $\gamma \geq 1$ and $1>\delta > 0$. If the new O-P is accepted, the linearized \optref{opt_LMI} is ``trusted'' with a wider exploration. Otherwise, the trust-region is shrunk by \eqref{eq_alpha_false}, i.e. it is searched more locally. A suitable choice of $\gamma$ and $\delta$ can prevent the algorithm from being attracted to an unacceptable local minimum in the convergence process.

We summarize the proposed procedure in
\begin{algorithm}%[H]
	\caption{Path-Following (executed off-line)}
	\label{algo_PF}
%%%%	algpseudocode	%%%%%%%%%%%%%%%%%%%%%%%%%%%%%%%%%%%%%%%%%%%%%%%%%%%%%%%%%%%%%%
	\begin{algorithmic}
		\State{\textbf{Define:} $\underline\Delta_\text{rel}>0$,
		 $k_\text{max}\in\mathbb{N}^+$,} \Comment{stopping criteria}
	 \State{\hspace{.475in}$\alpha_\text{max}>0$, $\gamma \geq 1$, $1>\delta > 0$} \Comment{adaptation setup}
	 \State{\textbf{Find} $\Q^0\succ\zero$ \textbf{such that:}} \Comment{initial operating point (O-P)}
		\State{\hspace{.2in}\optref{opt_BMI} with constraint $\Q = \Q^0$ is feasible, returns $\mPi^0$}
		\State{\textbf{Initialize:} $k=1$, $\Delta_\text{rel}^0=\underline\Delta_\text{rel}$, $\alpha^0 = 0.2$}\Comment{cf. \cite{Ostertag2008}}
		\While{$k \leq k_\text{max} \wedge   \Delta_\text{rel}^{k-1} \geq\underline\Delta_\text{rel}$}\Comment{cf. \cite{Ostertag2008}} 
			\State{\textbf{Solve:} \optref{opt_LMI}, returns $\Delta\Q^*$} \Comment{linearized OP at O-P $k-1$}
			\State{\textbf{Solve:} \optref{opt_BMI} under constraint $\Q=\Q^{k-1}+\Delta\Q^*$,} \State{\hspace{0.95cm}returns $\mPi^k$, $\mGamma^k$} \Comment{Is original OP feasible?}
			\If{feasible $\wedge \text{ } \Delta_\text{rel}^{k} >0$} \Comment{feasible and improvement \text{\ding{51}}}
				\State{$\Q^{k}=\Q^{k-1}+\Delta\Q^*$}
				\State{$\alpha^{k}\leftarrow$ \eqref{eq_alpha_true}} \Comment{adaptation: explorate} 
			\Else \Comment{infeasible or no improvement \text{\ding{55}}}
				\State{$\Q^{k}=\Q^{k-1}$, $\mPi^{k}=\mPi^{k-1}$, $\Delta_\text{rel}^{k}=\Delta_\text{rel}^{k-1}$}
				\State{$\alpha^{k}\leftarrow$ \eqref{eq_alpha_false}} \Comment{adaptation: search more locally}
			\EndIf
		\State{$k\leftarrow k+1$}
		\EndWhile
		\State{\textbf{Return} $\Q^*=\Q^k$, $\mPi^*=\mPi^{k}$, $\mGamma^*=\mGamma^k$} \Comment{EBOSS}
	\end{algorithmic}
\end{algorithm}

\remark Following \cite[Ch. 6]{Anderson2007}, a typical initialization $\Q^0=q\cdot\text{diag}\left(\tfrac{1}{\epsilon_1^2},\ldots,\tfrac{1}{\epsilon_p^2}\right)$ with suitable large $q>0$ should be sufficient to satisfy the constraints of \optref{opt_BMI} in most cases. In order to find a satisfying local minimum, however, trying different $\Q^0$ or several reinitializations may be necessary.
\section{SIMULATION RESULTS}
\label{sec_sim}

In this section we demonstrate our contributions C1-3). We show that optimal synchronization can lead to a satisfying performance even when standard approaches such as \cite{KimShimSeo2011} are infeasible. That is, for a given exosystem (\ref{eq_sshetMAS}d-e) the necessary solvability of the regulator equations \cite{WielandSepulchreAllgoewer2011} is violated and exact synchronization is impossible. In this regard, we consider an under-actuated agent for which $\rank{\B_i} < \rank{\C_i}$. Furthermore, we verify the energy-efficiency of our approach by comparing the energy-consumption of two homogeneous agents where one is affected by significant actuator wear.
\renewcommand{\arraystretch}{1.15}
\setlength{\columnsep}{1pt}
\begin{table}
	\vspace{0.04in}
	\caption{Synchro. strategy, color and property of each agent}
	\centering
	\begin{tabular}{@{\hspace{.1cm}}l@{\hspace{.3cm}}p{.9cm}@{\hspace{.3cm}}p{1.4cm}@{\hspace{.3cm}}p{.9cm}@{\hspace{.3cm}}p{.98cm}@{\hspace{.3cm}}p{1.4cm}@{\hspace{.1cm}}}
		\hline
		\rule{0pt}{10pt} Agent & \circledtext{A}\circledtext{1} & \circledtext{A}\circledtext{2} & \circledtext{B}\circledtext{3} & \circledtext{B}\circledtext{4} & \circledtext{B}\circledtext{5} \\
		\hline
		Synchro. & exact & EBOSS & exact & OSS & EBOSS \\
		Color & \agentred & \agentgold & \agentcyan & \agentgreen & \agentblue \\
		Property & -- & actuator-wear~$12.5\%$ & -- & over-actuated & under-actuated \\
		\hline
	\end{tabular}
	\vspace{-0.04in}
	\label{tab_syn}
\end{table}
\FloatBarrier
For this purpose, two groups of heterogeneous agents are considered. The first group reads
\begin{equation*}
 \text{\circledtext A\Bigg\{\,}\hspace{-.06in} \A_i = \ma -1 & 0 & 5 \\ 
 0 & 0 & 1\\ 
 -5 & 2 & 0 \me\hspace{-.06in}\text{, } \B_i = \beta_i \ma 2 & 2 \\ 0 & 0 \\ 1 & 2\me\hspace{-.06in}\text{, } \C_i= \ma 1 & 0 & 0 \\ 0 & 1 & 0 \me\hspace{-.06in}.
\end{equation*}
While agent \circledtext{1} is in healthy conditions, $\beta_1 = 1$, \circledtext{2} is subject to $12.5\%$ actuator wear, $\beta_2=0.875$. The second group is given by
\begin{equation*}
\text{\circledtext B\Bigg\{\,} \A_i = \ma 0 & 1 & 0 & 0\\
1 & 1 & 0 & 0\\ 0.5 & 1 & 0 & 1\\ 0 & 0.5 & 1.5 & 1 \me \hspace{-.06in}\text{, }\C_i = \ma 1 & 0 & 0 & 0 \\ 0 & 0 & 2 & 0 \me
\end{equation*}
with input matrices
\begin{equation*}
\B_3 = \ma 0 & 0\\ 1 & 0\\ 0 & 1\\ 0 & 0 \me\hspace{-.06in}\text{, }\B_4 = \ma1 & 0 & 0\\ 0 & 1 & 0\\ 0 & 0 & 1\\ 0 & 0 & 0 \me\hspace{-.06in}\text{, }\B_5 = \ma0\\ 1\\ 0\\ 0 \me\hspace{-.06in}.
\end{equation*}
Clearly, we have quadratic agent \circledtext{3}, over-actuated \circledtext{4} and under-actuated \circledtext{5}. The communication is organized in a ring topology: $\rightarrow$\circledtext{1}$\rightarrow$\circledtext{2}$\rightarrow\ldots\rightarrow$\circledtext{5}$\rightarrow$. \assref{assump_spanning_tree} and \ref{assump_ctrb_obsv} apparently hold. All agents are stabilized by linear-quadratic regulators with group-wise similar weightings.

The homogeneous exosystem \eqref{eq_hom_sys} of each agent has the frequency spectrum $\Ool=\{0,0,0.5,2\}$ with $N_{\Ool}=2$ for which \assref{assump_eig_Aol} and \ref{assump_freq_ratio} hold. The time period is $T\approx \SI{12.6}{\second}$. With the multiplicities $m_{\Ool}(0)=2$, $m_{\Ool}(0.5)=1$ and $m_{\Ool}(2)=1$, we have $L=2$, $H=2$. The exosystem is of order $\nol=L+2H=6$ and the output matrix is given by
\begin{equation}\label{eq_Col}
\Col = \ma 1 & 0 & 0 & 1 & 0 & 0 \\ 0 & 1 & 0.2 & 1 & 1 & 0 \me.
\end{equation}
The desired maximal step-heights and amplitudes are defined by $\P=\ma 2.5 & 1.5625 & \tfrac{1}{2}\I_2 & \tfrac{1}{4}\I_2  \me$ as in \eqref{eq_P}. With $\Bol=\I_{\nol}$, the synchronization gain $\Kol$ is obtained from \lemref{lem_homsync} for $\sigma = 0.138$. 

It is left to choose $(\mPi_i,\mGamma_i)$ applying for the individual local synchronization strategy of each agent. \tabref{tab_syn} gives an overview. In case of \circledtext{1} and \circledtext{3}, the classical approach of exact synchronization (EXS) is obtained with $(\mPi_{1/3},\mGamma_{1/3})$ solving the regulator equations \cite{Trentel2001}. The over-actuated \circledtext{4} synchronizes optimally based on \defref{def_syn}.\ref{def_opt_syn}) with $\Q_4=\text{diag}(30,20)$ and $\R_4=\I$. Then, the optimal pair $(\mPi_4^*,\mGamma_4^*)$ can be calculated as in \thmref{thm_opt_syn} right away. Apparently, the over-actuation does not need to be considered explicitly. This is beneficial in comparison to the classical EXS approach. There, the solution of the regulator equations would not be unique and one would have to solve an OP such as given in \cite{Krener1992} which additionally may lead to suboptimal solutions.

Due to the actuator wear, \circledtext{2} desires to save as much input-energy as possible while synchronizing within acceptable bounds $\epsilon_{21}=0.37$ and $\epsilon_{22}=0.28$. With $\R_2=\I$, \optref{opt_BMI} is numerically solved by means of PENLAB which gives diagonal $\Q_2^*=\text{diag}(448.47,391.83)$ and $\trace{\mGamma_2^{*^\tp}\R_2\mGamma_2^*\P^2}=260.68$. Implementing $(\mPi_2^*,\mGamma_2^*)$ leads to an error-bounded optimal synchronization based on \defref{def_syn}.\ref{def_opt_err_syn}).
\begin{figure}
	\vspace{0.04in}
	\centering
%	\tikzsetfigurename{outputs} 
	
\definecolor{mycolor1}{rgb}{1.00000,0.55000,0.00000}% %orange
\definecolor{mycolor2}{rgb}{0.40000,0.70000,0.40000}%
\definecolor{mycolor3}{rgb}{1.00000,0.84000,0.00000}% % gold
\definecolor{mycolor4}{rgb}{0.46667,0.53333,0.60000}% %gray
\begin{tikzpicture}[scale=0.665, every node/.style={scale=1.1}]

\begin{axis}[%
width=4.521in,
height=1.493in,
at={(0.758in,2.554in)},
scale only axis,
unbounded coords=jump,
xmin=0,
xmax=20,
ymin=0,
ymax=3.47,
ylabel style={font=\color{white!15!black}},
ylabel={$y_1$},
axis x line*=bottom,
axis y line*=left,
xmajorgrids,
ymajorgrids
]
\addplot [color=gray, dashed, line width=1.5pt, forget plot, draw opacity=0.7]
table[]{figures/outputs-8.tsv};

\addplot [color=dark_cyan, line width=1.5pt, forget plot]
table[]{figures/outputs-2.tsv};

\addplot [color=red, line width=1.5pt, forget plot]
  table[]{figures/outputs-1.tsv};

\addplot [color=mycolor2, line width=1.5pt, forget plot]
  table[]{figures/outputs-3.tsv};
  
\addplot[color=blue, line width=1.5pt, forget plot]
  table[]{figures/outputs-5.tsv};
  
  % circles:
 \addplot[color=blue, mark=o,mark options={solid}, line width=1.5pt, only marks, forget plot]	
  table[]{figures/outputs-5A.tsv};
  
\addplot[color=mycolor3, line width=1.5pt, forget plot]
  table[]{figures/outputs-4.tsv};
  
  % triangles:
  \addplot[color=mycolor3, mark=triangle*,mark options={solid}, line width=1.5pt, only marks, forget plot]	
  table[]{figures/outputs-4A.tsv};

\addplot[thick,color = gray, fill=gray, fill opacity=0.3, domain=0:20,samples=100, draw opacity=0.3] table[]{figures/outputs-6.tsv}
-- (axis cs:20,3.47) -- (axis cs:0,3.47);  
\addplot[thick,color = gray, fill=gray, fill opacity=0.3, domain=0:20,samples=100, draw opacity=0.3] table[]{figures/outputs-7.tsv} \closedcycle;
\end{axis}

\begin{axis}[%
width=4.521in,
height=1.493in,
at={(0.758in,0.8in)},
scale only axis,
xmin=0,
xmax=20,
xlabel style={font=\color{white!15!black}},
xlabel={time in s},
ymin=0,
ymax=2.69,
ylabel style={font=\color{white!15!black}},
ylabel={$y_2$},
axis background/.style={fill=white},
axis x line*=bottom,
axis y line*=left,
xmajorgrids,
ymajorgrids,
legend style={legend cell align=left, align=left, draw=white!15!black}
]
\addplot [color=gray, dashed, line width=1.5pt, forget plot, draw opacity=0.7]
table[]{figures/outputs-16.tsv};

\addplot [color=dark_cyan,line width=1.5pt]
table[]{figures/outputs-10.tsv};

\addplot [color=red,line width=1.5pt]
  table[]{figures/outputs-9.tsv};

\addplot [color=mycolor2,line width=1.5pt]
  table[]{figures/outputs-11.tsv};

\addplot [color=blue,line width=1.5pt]
  table[]{figures/outputs-13.tsv};

% triangles:
\addplot[color=blue, mark=o,mark options={solid}, line width=1.5pt, only marks, forget plot]	
table[]{figures/outputs-13A.tsv};

\addplot [color=mycolor3,line width=1.5pt]
table[]{figures/outputs-12.tsv};

% triangles:
\addplot[color=mycolor3, mark=triangle*,mark options={solid}, line width=1.5pt, only marks, forget plot]	
table[]{figures/outputs-12A.tsv};

  \addplot[thick,color = gray, fill=gray, fill opacity=0.3, domain=0:20,samples=100, draw opacity=0.3] table[]{figures/outputs-14.tsv}
  -- (axis cs:20,2.69) -- (axis cs:0,2.69);  
  \addplot[thick,color = gray, fill=gray, fill opacity=0.3, domain=0:20,samples=100, draw opacity=0.3] table[]{figures/outputs-15.tsv} \closedcycle;
\end{axis}
\end{tikzpicture}%
	\vspace{-0.1in}
	\caption{Outputs of the heterogeneous agents (cf. \tabref{tab_syn})}
	\vspace{-0.01in}
	\vspace{-0.04in}
	\label{fig_outputs}
\end{figure}
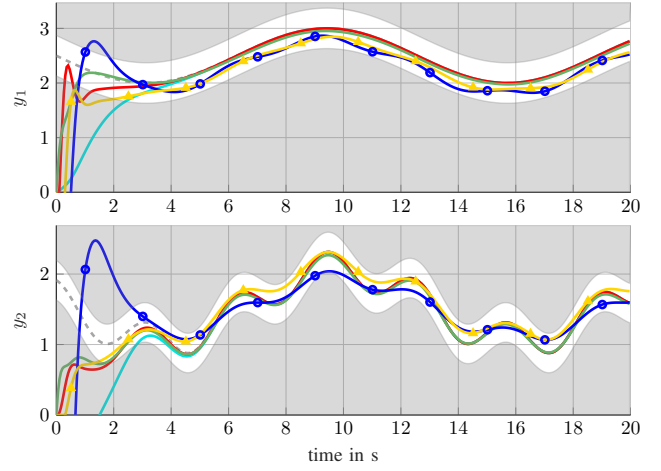
\FloatBarrier
In case of the under-actuated \circledtext{5}, an EXS solution does not exist since \assref{assump_EXS} is violated.
In an iterative manner following \remref{rem_eboss_it}, $\epsilon_{51}=1.7$, $\epsilon_{52}=2.4$ were obtained for which \optref{opt_BMI} with $\R_5=1$ can be solved. Starting at $Q_5^0=\text{diag}(15,16)$, the path following\footnote{The LMI-problems were numerically solved by the help of CVX \cite{Grant2013}.} \algoref{algo_PF} was carried out for $\gamma = 2.5$, $\delta = 0.5$, $\alpha_{\text{max}}=10$. It terminated after $29$ iterations due to $\Delta_\text{rel}^{29} <10^{-4}$. This resulted in {\small$Q_5^{*}=\ma10.8743 & -0.7571\\
-0.7571 &  12.1562\me$} and {\small$\trace{\mGamma_5^{*^\tp}\R_2\mGamma_5^*\P^2}=3.62$}. At this point, we stress that the bounds $\epsilon_{51}=1.7$, $\epsilon_{52}=2.4$ define a worst-case synchronization error. We will see next that even for $\xol(0)$ on the boundary of $\set{\overline X}$ the performance can be quite satisfying.

In the sequel, a simulation example is analyzed. At $t=0$, all agents are at rest and the exosystems are asynchronous such that  $\xol(0)^\tp=\ma 1 & 1 & 1 & 0 & 1 & 0  \me\cdot\P$. Although this requires $\xol_i(0) \not\in\set{\overline X}$ for some $i$, it results $\xol(0)\in\set{\overline X}$. Hence, \circledtext{2} and \circledtext{5} will satisfy the defined bounds and the results of \thmref{thm_opt_err_syn} hold. Then, the synchronization trajectory (\synchtraj) is given by $\yol(t)=\Col \xol(t)$ with \eqref{eq_Col} and
\begin{equation*}
\xol(t)^{\hspace*{-.5mm}\tp} \hspace*{-1mm}=\hspace*{-1mm} \ma 1 & \hspace{-.05in}1 & \hspace{-.05in}\cos(0.5t) & \hspace{-.05in}-\sin(0.5t) & \hspace{-.05in}\cos(2t) & \hspace{-.05in}-\sin(2t) \me\cdot\P.
\end{equation*}
This also shows that the exosystem formulation as required by \eqref{eq_Aol} is rather intuitive.

The results for outputs $y_1$ and $y_2$ are presented in \figref{fig_outputs}. After a transition period $[\SI{0}{\second},\SI{6}{\second}]$, we observe that all agents satisfy the error-bounds $\epsilon_{21}$ and $\epsilon_{22}$, i.e. the stationary trajectories omit the gray area. As expected, \circledtext{1} and \circledtext{3} track $\yol(t)$ asymptotically. Though $\xol(0)$ lies on the boundary of $\set{\overline X}$, the bounds are satisfied by \circledtext{2}. The over-actuated \circledtext{4} tracks $\yol(t)$ very closely while the under-actuated \circledtext{5} also shows a satisfying tracking performance.

The $i$-th agent's stationary input-energy over a period is {\small$J_{u,i}=\tfrac{1}{2}\int_{0}^{T}\xol^\tp\mGamma_i^\tp\R_i\mGamma_i\xol\text{ d}t$}. For group \circledtext{B}, the exact synchronizing \circledtext{3} with {\small$\R_3=\I$} is even less efficient than the under-actuated \circledtext{5}, i.e. $J_{u,3}=41.2>J_{u,5}=35.2>J_{u,4}=21.8$. The energy-consumption of group \circledtext{A} is displayed in \figref{fig_input_energy}. Since \circledtext{2} is not forced to synchronize exactly, it manages $J_{u,2}\approx J_{u,1}\approx760$ despite the $12.5\%$ actuator wear. In case of EXS, unfavorably, additional $32\%$ input-energy would have been necessary, see graph \agentdashgold. In the same manner, \circledtext{1} would have saved $24.4\%$ input-energy, if he had relaxed his synchronization to the acceptable bounds, see graph \agentdashred.
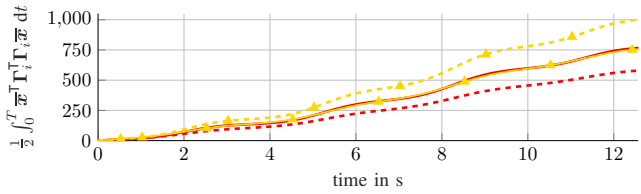
\begin{figure}
	\vspace{0.04in}
	\centering
%	\tikzsetfigurename{input_energy}  
\definecolor{mycolor1}{rgb}{1.00000,0.84000,0.00000}%
\begin{tikzpicture}[scale=0.665, every node/.style={scale=1.1}]

\begin{axis}[%
width=4.251in,
height=1.0in,
at={(0.758in,0.453in)},
scale only axis,
xmin=0,
xmax=12.5663706143592,
xlabel style={font=\color{white!15!black}},
xlabel={time in s},
ymin=0,
ymax=1050,
ytick = {0,250,500,750,1000},
ylabel style={font=\color{white!15!black}},
ylabel={$\tfrac{1}{2}\int_{0}^{T} \xol^\tp\mGamma_i^\tp\mGamma_i\xol\text{ d}t$},
axis background/.style={fill=white},
axis x line*=bottom,
axis y line*=left,
xmajorgrids,
ymajorgrids,
legend style={at={(0.157,0.469)}, anchor=south west, legend cell align=left, align=left, draw=white!15!black}
]

\addplot [color=red,line width=1.5pt]
  table[]{figures/input_energy-1.tsv};\label{energy_red}

\addplot [color=red, dashed,line width=1.5pt]
  table[]{figures/input_energy-2.tsv};

\addplot [color=mycolor1,line width=1pt]
table[]{figures/input_energy-3.tsv};

\addplot [color=mycolor1,mark=triangle*,line width=1.5pt,only marks]
table[]{figures/input_energy-3A.tsv};

\addplot [color=mycolor1, dashed,line width=1.5pt]
  table[]{figures/input_energy-4.tsv};

\addplot [color=mycolor1,mark=triangle*,line width=1.5pt,only marks]
table[]{figures/input_energy-4A.tsv};

\end{axis}
\end{tikzpicture}%
	\vspace{-0.125in}
	\caption{Stationary input-energy for group \circledtext{A} with $\R_1=\R_2=\I$}
	\vspace{-0.01in}
	\vspace{-0.04in}
	\label{fig_input_energy}
\end{figure}
\addtolength{\textheight}{-9.5cm} %-12
\section{CONCLUSION}
\label{sec_concl}
We presented an LQT-based approach for optimal stationary synchronization which can be considered an alternative to exact synchronization. Comparing both, the control structure is completely the same. However, our method shows various advantages. Typical assumptions for the existence of the control could be relaxed and the class of MAS suited for application is extended. It was shown that synchronization within acceptable bounds of the synchronization error allows saving a significant amount of input-energy. Hence, the agents achieve locally an optimized performance. Due to the modularity of our approach, the results can be recommended for application to general infinite-time LQT-problems.


%%%%%%%%%%%%%%%%%%%%%%%%%%%%%%%%%%%%%%%%%%%%%%%%%%%%%%%%%%%%%%%%%%%%%%%%%%%%%%%%

%%%%%%%%%%%%%%%%%%%%%%%%%%%%%%%%%%%%%%%%%%%%%%%%%%%%%%%%%%%%%%%%%%%%%%%%%%%%%%%%

% \section*{ACKNOWLEDGMENT}

%%%%%%%%%%%%%%%%%%%%%%%%%%%%%%%%%%%%%%%%%%%%%%%%%%%%%%%%%%%%%%%%%%%%%%%%%%%%%%%%

%\bibliographystyle{IEEEtran}
\bibliographystyle{IEEEtranS} % sorts based on names of authors
\bibliography{literature}

% Generated by IEEEtranS.bst, version: 1.14 (2015/08/26)
\begin{thebibliography}{10}
\providecommand{\url}[1]{#1}
\csname url@samestyle\endcsname
\providecommand{\newblock}{\relax}
\providecommand{\bibinfo}[2]{#2}
\providecommand{\BIBentrySTDinterwordspacing}{\spaceskip=0pt\relax}
\providecommand{\BIBentryALTinterwordstretchfactor}{4}
\providecommand{\BIBentryALTinterwordspacing}{\spaceskip=\fontdimen2\font plus
\BIBentryALTinterwordstretchfactor\fontdimen3\font minus
  \fontdimen4\font\relax}
\providecommand{\BIBforeignlanguage}[2]{{%
\expandafter\ifx\csname l@#1\endcsname\relax
\typeout{** WARNING: IEEEtranS.bst: No hyphenation pattern has been}%
\typeout{** loaded for the language `#1'. Using the pattern for}%
\typeout{** the default language instead.}%
\else
\language=\csname l@#1\endcsname
\fi
#2}}
\providecommand{\BIBdecl}{\relax}
\BIBdecl

\bibitem{Anderson2007}
B.~D.~O. Anderson and J.~B. Moore, \emph{Optimal Control: Linear Quadratic
  Methods}.\hskip 1em plus 0.5em minus 0.4em\relax Dover Publications, Inc.,
  2007.

\bibitem{Bernhard2016}
S.~Bernhard and J.~Adamy, ``Minimized input-energy gain based static decoupling
  control for linear over-actuated systems with sinusoidal references,'' in
  \emph{12th IEEE International Conf. on Control and Automation (ICCA)}, 2016,
  pp. 253--259.

\bibitem{Bernhard2017b}
S.~Bernhard, ``Time-invariant control in {LQ} optimal tracking: An alternative
  to output regulation,'' \emph{IFAC-PapersOnLine}, vol.~50, no.~1, pp.
  4912--4919, 2017.

\bibitem{Boyd1994}
S.~Boyd, L.~E. Ghaoui, E.~Feron, and V.~Balakrishnan, \emph{Linear Matrix
  Inequalities in System and Control Theory}.\hskip 1em plus 0.5em minus
  0.4em\relax SIAM, 1994.

\bibitem{Fiala2013}
J.~Fiala, M.~Ko\v{c}vara, and M.~Stingl, ``{PENLAB}: A {MATLAB} solver for
  nonlinear semidefinite optimization,'' 2013, [Online]
  \url{https://arxiv.org/abs/1311.5240}.

\bibitem{Grant2013}
M.~Grant and S.~Boyd, ``{CVX}: Matlab software for disciplined convex
  programming, version 2.0 beta,'' September 2013, [Online]
  \url{http://cvxr.com/cvx}.

\bibitem{Halkin1974}
H.~Halkin, ``Necessary conditions for optimal control problems with infinite
  horizons,'' \emph{Econometrica}, vol.~42, pp. 267--272, 1974.

\bibitem{Hermann2018}
J.~Hermann, S.~Bernhard, U.~Konigorski, and J.~Adamy, ``Designing communication
  topologies for optimal synchronization trajectories of homogeneous linear
  multi-agent systems,'' in \emph{European Control Conference (ECC)}, 2018.

\bibitem{Khodaverdian2014-ifacWC}
S.~Khodaverdian and J.~Adamy, ``Synchronizing linear heterogeneous networks by
  output homogenization,'' in \emph{19th IFAC World Congress}, 2014, pp.
  4687--4692.

\bibitem{Khodaverdian2015-cdc}
------, ``Distributed dynamic decoupling-based output synchronization for
  networks of linear heterogeneous {MIMO} agents,'' in \emph{54th IEEE Conf. on
  Decision and Control}, 2015, pp. 6202--6208.

\bibitem{KimShimSeo2011}
H.~Kim, H.~Shim, and J.~H. Seo, ``Output consensus of heterogeneous uncertain
  linear multi-agent systems,'' \emph{IEEE Transactions on Automatic Control},
  vol.~56, no.~1, pp. 200--206, 2011.

\bibitem{Kreindler1969}
E.~Kreindler, ``On the linear optimal servo problem,'' \emph{International
  Journal of Control}, vol.~9, no.~4, pp. 465--472, 1969.

\bibitem{Krener1992}
A.~J. Krener, ``The construction of optimal linear and nonlinear regulators,''
  \emph{Systems, Models and Feedback: Theory and Applications}, vol.~12, pp.
  301--322, 1992.

\bibitem{LafferriereWilliamsCaughmanVeerman2005}
G.~Lafferriere, A.~Williams, J.~Caughman, and J.~J.~P. Veerman, ``Decentralized
  control of vehicle formations,'' \emph{Systems \& Control Letters}, vol.~54,
  no.~9, pp. 899--910, 2005.

\bibitem{MaZhang2010}
C.-Q. Ma and J.-F. Zhang, ``Necessary and sufficient conditions for
  consensusability of linear multi-agent systems,'' \emph{IEEE Transactions on
  Automatic Control}, vol.~55, no.~5, pp. 1263--1268, 2010.

\bibitem{Montenbruck2015}
J.~M. Montenbruck, M.~B\"urger, and F.~Allg\"ower, ``Practical synchronization
  with diffusive couplings,'' \emph{Automatica}, vol.~53, pp. 235--243, 2015.

\bibitem{Olfati-SaberFaxMurray2007}
R.~Olfati-Saber, J.~A. Fax, and R.~M. Murray, ``Consensus and cooperation in
  networked multi-agent systems,'' \emph{Proceedings of the IEEE}, vol.~95,
  no.~1, pp. 215--233, 2007.

\bibitem{Ostertag2008}
E.~Ostertag, ``An improved path-following method for mixed ${H}_2$ /
  ${H}_\infty$ controller design,'' \emph{IEEE Transactions on Automatic
  Control}, vol.~53, pp. 1967--1971, 2008.

\bibitem{Peymani2014}
E.~Peymani, H.~F. Grip, A.~Saberi, X.~Wang, and T.~I. Fossen, ``{$\mathcal
  H_\infty$} almost output synchronization for heterogeneous networks of
  introspective agents under external disturbances,'' \emph{Systems \& Control
  Letters}, vol.~50, no.~4, pp. 1026--1036, 2014.

\bibitem{RenBeardAtkins2007}
W.~Ren, R.~W. Beard, and E.~M. Atkins, ``Information consensus in multivehicle
  cooperative control,'' \emph{IEEE Control Systems Magazine}, vol.~27, no.~2,
  pp. 71--82, 2007.

\bibitem{Shim2015}
H.~Shim and S.~Trenn, ``A preliminary result on synchronization of
  heterogeneous agents via funnel control,'' in \emph{54th IEEE Conf. on
  Decision and Control}, 2015, pp. 2229--2234.

\bibitem{Trentel2001}
H.~Trentelman, A.~A. Stoorvogel, and M.~Hautus, \emph{Control Theory for Linear
  Systems}.\hskip 1em plus 0.5em minus 0.4em\relax Springer-Verlag London,
  2001.

\bibitem{Tuna2008}
S.~E. Tuna, ``{LQR}-based coupling gain for synchronization of linear
  systems,'' 2008, [Online] \url{https://arxiv.org/abs/0801.3390}.

\bibitem{Antwerp2000}
J.~G. VanAntwerp and R.~D. Braatz, ``A tutorial on linear and bilinear matrix
  inequalities,'' \emph{Journal of Process Control}, vol.~10, pp. 363--385,
  2000.

\bibitem{WielandSepulchreAllgoewer2011}
P.~Wieland, R.~Sepulchre, and F.~Allg\"{o}wer, ``An internal model principle is
  necessary and sufficient for linear output synchronization,''
  \emph{Automatica}, vol.~47, no.~5, pp. 1068--1074, 2011.

\bibitem{Willems2004}
J.~L. Willems and I.~M.~Y. Mareels, ``A rigorous solution of the infinite time
  interval {LQ} problem with constant state tracking,'' \emph{System \& Control
  Letters}, vol.~52, no. 3-4, pp. 289--296, 2004.

\bibitem{YangSaberiStoorvogelGrip2014}
T.~Yang, A.~Saberi, A.~A. Stoorvogel, and H.~F. Grip, ``Output synchronization
  for heterogeneous networks of introspective right-invertible agents,''
  \emph{International Journal of Robust and Nonlinear Control}, vol.~24,
  no.~13, pp. 1821--1844, 2014.

\end{thebibliography}

\end{document}